\documentclass[]{article}

\usepackage{amssymb,amsbsy,amsthm,amsmath, amssymb, graphicx}
\usepackage{eepic,texdraw,graphics,tikz}

\usepackage{fourier}

\newtheorem{theorem}{Theorem}[section]
\newtheorem{lemma}[theorem]{Lemma}
\newtheorem{proposition}[theorem]{Proposition}
\newtheorem{corollary}[theorem]{Corollary}
\newtheorem{definition}[theorem]{Definition}

\theoremstyle{remark} %for plain letters 
\newtheorem{remark}{Remark}[section]
\newtheorem{example}{Example}[section]

\newcommand{\mathsym}[1]{{}}
\newcommand{\unicode}[1]{{}}

\title{On consistent systems of difference equations}
\author{Pavlos Xenitidis\footnote{e-mail: xenitip@hope.ac.uk} \\ Department of Mathematics \& Computer Science\\ Liverpool Hope University,  L16 9JD Liverpool, UK}
\date{\today}

\begin{document}

\maketitle

\begin{abstract}
We consider overdetermined systems of difference equations for a single function $u$ which are consistent, and propose a general framework for their analysis. The integrability of such systems is defined as the existence of higher order symmetries in both lattice directions and various examples are presented. Two hierarchies of consistent systems are constructed, the first one using lattice paths and the second one as a deformation of the former. These hierarchies are integrable and their symmetries are related via Miura transformations to the Bogoyavlensky and the discrete Sawada-Kotera lattices, respectively.
\end{abstract}

\section{Introduction}
Difference equations defined on an elementary quadrilateral of the lattice, also referred to as quad equations, constitute probably the most well known and well studied class of discrete integrable systems, see for instance \cite{HJN} and references therein. Their integrability can be established in various ways and the most rigorous one is provided by the existence of infinite hierarchies of generalized symmetries in both lattice directions, i.e. evolution type differential-difference equations compatible with them. 

Even though integrable quad equations admit only one hierarchy of symmetries in one direction, the same hierarchy may also be compatible with $N$-quad equations, i.e. difference equations  defined on $N>1$ consecutive quadrilaterals on the lattice, \cite{AP1,BMX,X3}. More interestingly, such differential-difference equations may also define symmetries of {\emph{overdetermined systems of difference equations}} which are {\emph{consistent}} \cite{MX}. For the continuous case the notion of consistent systems of hyperbolic type was introduced in \cite{AS}. In the discrete case, there exist sporadic examples of consistent systems which suggest that they could be related to a quad equation \cite{NRGO}, or follow from quad equations via potentiation \cite{MX}, or even from the degeneration of symmetries of two-quad equations as we demonstrate below. But there do exist integrable consistent systems which cannot be derived from a scalar equation in any of the aforementioned ways.

Here we consider first of all consistent systems which involve two two-quad equations, or, in our terminology, consistent systems of order two. We discuss their properties and symmetries and their relation to quad equations. Motivated by these examples, we propose a general framework for consistent systems of any order, and analyse the stencil on which they are defined. We discuss certain choices for dynamical variables and how they are related to the initial value problem. In particular the so-called standard dynamical variables are closely related to the symmetries of the system, and thus to its integrability.

We construct two novel hierarchies of consistent systems and discuss their integrability properties. The first hierarchy is constructed using a nice and simple method which employs lattice paths connecting the origin with the lattice points $(i,N+1-i)$, with $i=1,\ldots,N$. The integrability of the members of this hierarchy is established by the derivation of the lowest order symmetries in both lattice directions which are related to the Bogoyavlensky lattice.  The construction of the second hierarchy is more involved and only two systems were constructed explicitly. Their symmetries are given and it is shown that they are related to the Sawada-Kotera lattice. Moreover it is shown how one hierarchy can be viewed as a deformation of the other, and how these hierarchies generalise two well-known quad equations, namely equation 
$$u_{n,m} (u_{n+1,m}+u_{n,m+1}) u_{n+1,m+1} + \alpha = 0$$
derived in \cite{MX}, and Adler's Tzitzeica equation studied in \cite{A}, 
$$u_{n,m} (u_{n+1,m} + u_{n,m+1}) u_{n+1,m+1} + c = u_{n,m} + u_{n+1,m+1} + \frac{u_{n,m} u_{n+1,m} u_{n,m+1} u_{n+1,m+1}}{c}.$$
In this way we establish that these two quad equations are not some isolated objects but the lowest order members of two integrable hierarchies of consistent systems. 

The paper is organised as follows. In Section \ref{sec:2q} we consider some examples in order to explore certain connections of quad and two-quad equations with consistent systems and analyse the properties of the latter. The following section is devoted to the development of a framework for the study and analysis of consistent systems of higher order, whereas Section \ref{sec:sys} deals with the derivation of two hierarchies of consistent systems, the study of their properties, and the analysis of their relation. Finally, in the concluding section we discuss various perspectives on the subject.

\section{From scalar equations to consistent systems} \label{sec:2q}

In this section we introduce our notation and give some necessary definitions in order to make our presentation self-contained. Then we consider overdetermined systems and check whether they are consistent or not. We discuss how such systems can be derived from scalar equations and finally we present a systematic method for their construction starting with a two-quad equation and its lowest order symmetry.

Throughout this paper we deal with autonomous partial difference equations, or systems thereof, involving one unknown  function $u$ of two independent discrete variables $n$ and $m$. Since $n$, $m$ do not appear explicitly in any of our systems, we can, without loss of generality, present all equations evaluated at $n=m=0$. Therefore, in what follows we use the notation $u_{i,j}$ to denote the value of $u$ at the lattice point $(i,j)$, i.e. $u_{i,j} = u(i,j)$. Moreover, ${\mathcal{S}}$ and ${\mathcal{T}}$ will denote the shift operators in the first and the second direction, respectively, defined as ${\mathcal{S}}^i{\mathcal{T}}^j(u_{0,0}) = u_{i,j}$. 

With a symmetry of a system of partial difference equations we mean an evolution type differential-difference equation compatible with the discrete system. More precisely,
\begin{definition}
Let  $u$ depend also on a continuous variable $t$. Then, the differential-difference equation
$$\partial_t u_{0,0} =F( [u])$$
defines a symmetry of the system of difference equations $Q([u])=0$ if 
$$\sum_{i,j} \frac{\partial Q}{\partial u_{i,j}} {\cal{S}}^i {\cal{T}}^j(F) = 0 $$
holds on solutions of the system. Here, the notation $[u]$ means that these functions depend on a finite but otherwise unspecified number of shifts of $u$.
\end{definition} 

We exemplify the notion of consistency with the use of two examples of systems involving two two-quad equations

\begin{example}
	Consider the overdetermined system
	\begin{subequations}\label{eq:nonint-sys}
		\begin{eqnarray}
		&& u_{0,0} (u_{1,0} u_{1,1} + u_{0,1} u_{0,2}) u_{1,2} - \alpha = 0,\label{eq:nonint-sys-E1} \\
		&& u_{0,0} (u_{1,0} u_{2,0} + u_{0,1} u_{1,1}) u_{2,1}  + \alpha = 0.\label{eq:nonint-sys-E2}
		\end{eqnarray}
	\end{subequations}
	In order to verify its consistency, first we write these equations as
	\begin{equation} \label{eq:noint-sys-1}
	u_{1,2}= \frac{\alpha}{u_{0,0} (u_{1,0} u_{1,1} + u_{0,1} u_{0,2})} ,\quad  u_{2,1} = \frac{-\alpha}{ u_{0,0} (u_{1,0} u_{2,0} + u_{0,1} u_{1,1})},
	\end{equation}
	and then check if the compatibility condition ${\cal{S}}\left(u_{1,2}\right) = {\cal{T}}\left(u_{2,1}\right)$ holds modulo system (\ref{eq:noint-sys-1}). Equivalently, we can check if the two different ways to compute $u_{2,2}$ lead to the same answer. If we shift the first equation in the first direction and then use (\ref{eq:noint-sys-1}) to eliminate $u_{1,2}$ and $u_{2,1}$, we will end up with
	$$ u_{2,2} = \frac{u_{0,0}}{u_{1,0} u_{0,1}} \frac{(u_{1,0} u_{2,0} +u_{0,1} u_{1,1}) (u_{0,1} u_{0,2} + u_{1,0} u_{1,1})}{u_{1,1}^2 - u_{2,0} u_{0,2}}.$$
	On the other hand, the shift of the second equation in the second direction and the use of the system for the elimination of $u_{1,2}$ and $u_{2,1}$ lead to the same expression for $u_{2,2}$. This clearly shows that the compatibility condition ${\cal{S}}\left(u_{1,2}\right) = {\cal{T}}\left(u_{2,1}\right)$ does not impose further restrictions on $u$, and therefore system (\ref{eq:nonint-sys}) is consistent. \hfill $\Box$
\end{example}

\begin{example}
	Another consistent system is the bilinear equations for the $\tau$-function of the lattice KdV given in \cite{NRGO},
	\begin{subequations} \label{eq:tau-H1}
		\begin{eqnarray}
		&& (p+q) \tau_{0,2} \tau_{1,0} - (p-q) \tau_{0,0} \tau_{1,2} - 2 q \tau_{0,1} \tau_{1,1} = 0, \label{eq:tau-H1-E1} \\
		&& (p+q) \tau_{2,0} \tau_{0,1} - (q-p) \tau_{0,0} \tau_{2,1} - 2 p \tau_{1,0} \tau_{1,1} = 0, \label{eq:tau-H1-E2}
		\end{eqnarray}
	\end{subequations}
	where $p,q \in {\mathbb{R}}$. It is a simple calculation to verify that this system is consistent and in particular to show that its consistency leads to 
	$$ \tau_{2,2} = \left( \frac{p+q}{p-q}\right)^2 \frac{\tau_{2,0} \tau_{0,2}}{\tau_{0,0}} - \frac{4 p q}{(p-q)^2} \frac{\tau_{1,1}^2}{\tau_{0,0}},$$
	i.e. a discrete Toda equation. \hfill $\Box$
\end{example}

Consistent systems are relatively rare and probably more difficult to construct. Such systems may follow from the potentiation of lower order systems as it is demonstrated in the following example. See also  \cite{MX} for other examples.

\begin{example}{{\emph{Potentiation of a quad equation}}}
	
	We start with equation \cite{MX}
	\begin{equation} \label{eq:quad-eq}
	v_{1,0} v_{0,1} \left(v_{0,0}+v_{1,1}\right) + 1 = 0
	\end{equation}
	and its conservation law 
	$$({\cal{T}}-1) \ln \frac{v_{0,0}}{v_{2,0} v_{1,0} v_{0,0}-1} = ({\cal{S}}-1) \ln (v_{0,0} v_{1,0}).$$
	We can use this conserved form of (\ref{eq:quad-eq}) to introduce a potential $u$ via the relations
	\begin{equation} \label{eq:intro-ex}
	\frac{u_{1,0}}{u_{0,0}}\,=\,\frac{v_{0,0}}{v_{2,0} v_{1,0} v_{0,0}-1}\,,\quad \frac{u_{0,1}}{u_{0,0}}\,=\,v_{0,0} v_{1,0}\,.
	\end{equation}
	If we solve them for $u_{1,0}$ and $u_{0,1}$, their compatibility condition ${\cal{T}}(u_{1,0}) = {\cal{S}}(u_{0,1})$ is identically zero on solutions of (\ref{eq:quad-eq}). On the other hand, it follows from the equations  that
	\begin{equation} \label{eq:intro-ex1}
	v_{1,0} = \frac{u_{0,1}}{u_{0,0} v_{0,0}},\quad v_{2,0} = \frac{u_{0,0} (u_{0,0} v_{0,0}+u_{1,0})}{u_{1,0} u_{0,1}}. 
	\end{equation}
	The compatibility condition ${\cal{S}}(v_{1,0}) = v_{2,0}$ of the latter system implies
	\begin{equation} \label{eq:intro-con}
	v_{0,0} = \frac{u_{1,0}}{u_{1,1}-u_{0,0}}.
	\end{equation}
	Substituting this into the first relation in (\ref{eq:intro-ex1}) and the quad equation (\ref{eq:quad-eq}) we end up with the system
	\begin{subequations} \label{eq:intro-sys}
		\begin{eqnarray}
		&& (u_{1,0}-u_{0,2}) (u_{0,0}-u_{1,1}) u_{0,1} - (u_{1,0}-u_{0,2}) u_{0,0} u_{1,2} - u_{0,2} u_{1,1} u_{1,2} =0, \label{eq:intro-sys-1}\label{eq:intro-sys-E1}\\
		&& (u_{1,0}-u_{2,1}) (u_{0,0}-u_{1,1}) u_{0,1} - u_{0,0} u_{1,0} u_{2,0}\,=\,0.\label{eq:intro-sys-E2}
		\end{eqnarray} 
	\end{subequations} 
	It can be shown that system (\ref{eq:intro-sys}) is consistent. Indeed, rearranging the equations of the system and write them as
	\begin{equation}\label{eq:intro-sol-sys}
	u_{1,2} =  \frac{ (u_{1,0}-u_{0,2}) (u_{0,0}-u_{1,1}) u_{0,1}}{ (u_{1,0}-u_{0,2}) u_{0,0} + u_{0,2} u_{1,1} },\quad u_{2,1} =  u_{1,0} - \frac{u_{0,0} u_{1,0} u_{2,0}}{(u_{0,0}-u_{1,1}) u_{0,1}},
	\end{equation}
	we can easily show that both of them lead to the same expression for $u_{2,2}$, namely
	$$ u_{2,2} = u_{0,0} \left( 1 - \frac{u_{2,0}}{u_{0,1}} - \frac{u_{1,0}}{u_{0,2}} + \frac{u_{0,0}}{u_{0,0}-u_{1,1}} \,\frac{u_{1,0} u_{2,0}}{u_{0,1}u_{0,2}}\right).$$
	Finally, using the Lax pair for (\ref{eq:quad-eq}) found in \cite{MX} along with relations (\ref{eq:intro-ex1}) and (\ref{eq:intro-con}) we end up with
	\begin{equation} \label{lax-pair}
	\Psi_{1,0}= \left( \begin{array}{ccc} 0 & 1 & 0\\ \tfrac{u_{1,0}}{u_{0,0}-u_{1,1}} & \tfrac{-u_{0,1}}{u_{0,0}} & \lambda \\ -1 & 0 & \tfrac{u_{1,1} -u_{0,0}}{u_{1,0}} \end{array} \right) \Psi_{0,0},\quad  \Psi_{0,1}= \left( \begin{array}{ccc} 0 & 0 & 1\\ -1 & 0 & \tfrac{u_{1,1} -u_{0,0}}{u_{1,0}} \\  \tfrac{u_{0,2} u_{1,1} + u_{0,0} (u_{1,0} -u_{0,2})}{\lambda u_{0,1} (u_{0,0}-u_{1,0})} &  \tfrac{-1}{\lambda} & 0 \end{array} \right)  \Psi_{0,0},
	\end{equation}
	which is a Lax pair for system (\ref{eq:intro-sys}). \hfill $\Box$
\end{example}

Consistent systems can be derived from a rather unusual approach employing symmetries. It is well known that symmetries provide us the means to find particular classes of solutions, aka group invariant solutions, by solving the overdetermined system of the equation and the vanishing of the characteristic of the symmetry generator. But if the symmetry generator is a rational expression, we may consider the vanishing of its denominator as an additional equation. This looks odd in first place but surprisingly it provides us with equations consistent with our original equation as it is explained in the following example. See also \cite{BMX} for quadrilateral equations  defining particular solutions of two-quad equations and \cite{AP} for examples involving higher order quad equations.

\begin{example}{{\emph{Degeneration of symmetries and consistency}}}\label{ex:deg}
	
	Consider the first equation of system (\ref{eq:intro-sys}) as a single two-quad equation,
	\begin{equation}\label{eq:intro-ex2-3}
	(u_{1,0}-u_{0,2}) (u_{0,0}-u_{1,1}) u_{0,1} - (u_{1,0}-u_{0,2}) u_{0,0} u_{1,2} - u_{0,2} u_{1,1} u_{1,2} =0.
	\end{equation}
	It is a straightforward but cumbersome calculation to show that the differential-difference equations 
	\begin{equation} \label{eq:intro-ex2-4}
	\partial_{t^\prime} u_{0,0} = \frac{u_{0,0} u_{1,0} u_{0,1} (u_{0,0}-u_{1,1})}{ (u_{0,0}-u_{1,1}) (u_{-1,0}-u_{0,1}) u_{-1,1} - u_{-1,0} u_{0,0} u_{1,0} } ,\quad \partial_{s} u_{0,0} \,=\,\frac{u_{0,0} u_{0,1} u_{0,2}}{(u_{0,2}-u_{0,-1}) (u_{0,1}-u_{0,-2})}
	\end{equation}
	define generalized symmetries of (\ref{eq:intro-ex2-3}). What is not so obvious is that if we shift the denominator of the first symmetry forward in the first direction and set it equal to zero, we will end up with
	\begin{equation}\label{eq:intro-ex2-1}
	(u_{1,0}-u_{2,1}) (u_{0,0}-u_{1,1}) u_{0,1} - u_{0,0} u_{1,0} u_{2,0} =0,
	\end{equation}
	which is consistent with (\ref{eq:intro-ex2-3}). In other words we could have derived consistent system (\ref{eq:intro-sys}) not as a potential form of (\ref{eq:quad-eq}) but starting with equation (\ref{eq:intro-ex2-3}) and requiring the degeneration of one of its symmetries.
	
	Alternatively, we could have considered equation (\ref{eq:intro-ex2-1}) and its generalized symmetries
	\begin{equation} \label{eq:intro-ex2-2}
	\partial_{t} u_{0,0} \,=\,u_{0,0}\,\left(\frac{u_{2,0}}{u_{-1,0}} + \frac{u_{1,0}}{u_{-2,0}} \right),\quad \partial_{s^\prime} u_{0,0} \, = \,\frac{u_{0,0} u_{0,-1} u_{1,-1}(u_{0,0}-u_{1,1})}{(u_{1,1}-u_{0,0}) (u_{0,1}-u_{1,-1}) u_{0,-1} + (u_{0,1}-u_{1,-1}) u_{0,0} u_{1,0} - u_{1,0} u_{0,1} u_{1,1}}.
	\end{equation}
	It is not difficult now to see that the denominator of the second symmetry shifted forward in the second direction is the defining function of (\ref{eq:intro-ex2-3}).  Thus we could have derived system (\ref{eq:intro-sys}) in two different ways without any reference to the quad equation (\ref{eq:quad-eq}). 
	
	A very interesting observation is that the lowest order symmetries of system (\ref{eq:intro-sys}) are given by the first flow in (\ref{eq:intro-ex2-2}) and the second one in (\ref{eq:intro-ex2-4}), i.e. by the symmetries of (\ref{eq:intro-ex2-3}) and (\ref{eq:intro-ex2-1}) which do not degenerate on the solutions of the overdetermined system (\ref{eq:intro-sys}).  \hfill $\Box$
\end{example}

Our last example is on the construction of a consistent system starting with a two-quad equation and its symmetry. This constructive approach will be used later in the derivation of a consistent system of order three.

\begin{example}{{\emph{Construction of a consistent system}}} \label{ex:const}
	
	We start with equation
	\begin{equation}\label{eq:ex2-e1}
	E_1 := u_{0,1} u_{0,2} \left(1+ a (u_{0,0}+u_{1,0})\right) + u_{1,0} u_{0,2} \left(1+ a u_{1,1}\right) + u_{1,0} u_{1,1} \left(1+ a u_{1,2}\right) = 0
	\end{equation}
	which possesses a generalised symmetry of order 3 in the second lattice direction generated by
	\begin{equation}\label{eq:ex2-sym}
	\partial_s u_{0,0} = u_{0,0} (1+ a u_{0,0}) (u_{0,3} u_{0,2} u_{0,1} - u_{0,-1} u_{0,-2} u_{0,-3}).
	\end{equation}
	Suppose that $E_2(u_{0,0},u_{1,0},u_{2,0},u_{0,1},u_{1,1},u_{2,1}) = 0$ is another equation consistent with (\ref{eq:ex2-e1}). If we shift it forward in the second direction, eliminate $u_{2,2}$ and $u_{1,2}$ using (\ref{eq:ex2-e1}) and its shift, then the resulting expression must independent of $u_{0,2}$. Thus, if we differentiate it with respect to $u_{0,2}$ and then shift backwards in the second direction, we will end up with
	$$ a u_{1,1} \left(a u_{0,1} (\partial_{u_{0,1}} E_2) + (1+ a u_{1,1} )  (\partial_{u_{1,1}} E_2)\right) + (1+a u_{1,1}) (1+a u_{2,1}) (\partial_{u_{2,1}} E_2) = 0,$$
	after the use of the backward shift of equation (\ref{eq:ex2-e1}) for the elimination of variables $u_{2,-1}$ and $u_{1,-1}$.
	
	On other hand if we use (\ref{eq:ex2-e1}) and its shift to eliminate $u_{2,0}$ and $u_{1,0}$ from $E_2$, then the resulting expression should be independent of $u_{0,0}$. Its differentiation with respect to $u_{0,0}$ yields
	$$ (1+a u_{1,0})  \left((1+ a u_{0,0}) (\partial_{u_{0,0}} E_2) + a u_{1,0} )  (\partial_{u_{1,0}} E_2)\right) + a^2 u_{1,0} u_{2,0} (\partial_{u_{2,0}} E_2)=0.$$
	These two linear partial differential equations are compatible and their solution is
	$$ E_2 = F(z_1,z_2,z_3,z_4) = F \left( \frac{u_{1,0}}{1+ a u_{0,0}}, \frac{u_{2,0}}{1+a u_{1,0}}, \frac{1+ a u_{1,1}}{u_{0,1}}, \frac{1+ a u_{2,1}}{ a u_{1,1}}\right) .$$
	
	Now we require $E_2=0$ to be consistent with the symmetry (\ref{eq:ex2-sym}), i.e. the determining  equation
	\begin{equation} \label{eq:ex2-deteq}
	\sum_{i=0}^{2} \sum_{j=0}^1  \left(u_{i,j} (1+ a u_{i,j}) (u_{i,j+3} u_{i,j+2} u_{i,j+1} - u_{i,j-1} u_{i,j-2} u_{i,j-3})\right) \left(\partial_{u_{i,j}} E_2\right) = 0
	\end{equation}
	must hold on solutions of $E_1 = E_2 = 0 $. We eliminate variables $u_{\ell,-3}$, $u_{\ell,-2}$, $u_{\ell,-1}$, $u_{\ell,2}$, $u_{\ell,3}$ and $u_{\ell,4}$ with $\ell =1,2$, from (\ref{eq:ex2-deteq}) using (\ref{eq:ex2-e1}) and its shifts. This results to an equation which apart from the variables appearing in the arguments of $E_2$ involves also $u_{0,-2}$, $u_{0,-1}$, $u_{0,2}$ and $u_{0,3}$. The coefficient of $u_{3,0}$ leads to
	$$ a z_1 z_2 (1+a z_1) F_{z_1} + z_2 (1+a z_2) F_{z_2} - a z_2 z_3 F_{z_3} - z_4 F_{z_4}=0.$$
	The general solution to this equation can be written as 
	$$ F(z_1,z_2,z_3,z_4) = G(t_1,t_2,t_3) \quad {\text{where}} \quad  t_1= \frac{z_1 z_3}{1+a z_1},\,\,t_2= \frac{1+a z_2 (1+a z_1)}{a z_1}, \,\,t_3 = \frac{a z_2 z_4 (1+a z_1)}{z_1}.$$
	In view of this, the coefficient of $u_{0,-2}$ becomes
	$$  (1+t_1+t_1 t_2) G_{t_2}  - (a-t_1 t_3) G_{t_3} + a z_1 \left( t_1 (1+t_1) G_{t_1} - t_2 G_{t_2} + a t_2 G_{t_3} \right) = 0,$$
	where $z_1$ plays the role of a separation variable. Solving this system for $G$ we end up with
	$$ E_2 = F(z_1,z_2,z_3,z_4) = G(t_1,t_2,t_3) = H\left(\frac{t_3+ a t_2 + t_1 t_3 }{1 + t_1 + t_1 t_2}\right) = H \left(\frac{a^2 z_1 z_2 (1+z_4) + a z_2 (1+z_4 +z_1 z_3 z_4) + 1}{z_1 (z_3 + a (1+z_2 z_3))} \right) = H(x). $$
	
	Finally, taking into account this form for $E_2$ and after the elimination of all variables as described above, the determining equation (\ref{eq:ex2-deteq}) can be written as $ x\, H^\prime(x)  = 0 $, which clearly implies that $H(x) = x$ and $x=0$ is the sought equation, or explicitly 
	\begin{equation}
	u_{2,0}\left(1+a u_{2,1}\right) \left\{ u_{1,0} \left(1+a u_{1,1}\right)   + u_{0,1}  \left(1+a (u_{0,0}+u_{1,0})\right)\right\}  + u_{0,1} u_{1,1} \left\{\left(1+ a u_{0,0}\right) \left(1+ a u_{1,0}\right) + a u_{2,0} \left(1+ a (u_{0,0}+u_{1,0})\right)\right\}=0. \label{eq:ex2-e2}
	\end{equation}
	It can be easily checked that equations (\ref{eq:ex2-e1}) and (\ref{eq:ex2-e2}) are consistent and (\ref{eq:ex2-sym}) is a symmetry of the system. 
	
	A symmetry in the first direction can be found using only equation (\ref{eq:ex2-e2}) and the method of \cite{X3} and can be written as
	\begin{subequations} \label{eq:sym-V}
		\begin{equation}
		\partial_t u_{0,0} = \frac{V_{0,0} p_{0,0}}{q_{0,0} q_{-1,0}} \left(\frac{V_{1,0} V_{2,0} p_{-1,0}}{q_{1,0}} - \frac{V_{-1,0} V_{-2,0} p_{1,0}}{q_{-2,0}} - r_{0,0}\right),
		\end{equation}
		where $V_{0,0} = u_{0,0} (1+\alpha u_{0,0})$ and
		\begin{eqnarray}
		q_{0,0} &=& (1+\alpha u_{0,0})  (1+\alpha u_{1,0})  (1+\alpha u_{2,0}) + \alpha u_{-1,0} \left(1+ \alpha(u_{0,0}+u_{1,0}+u_{2,0}) + \alpha^2 (u_{0,0} u_{1,0} + u_{1,0} u_{2,0} + u_{2,0} u_{0,0})\right),\\
		p_{0,0} &=& (1+\alpha u_{0,0})  (1+\alpha u_{1,0}) + \alpha u_{-1,0} \left(1+  \alpha(u_{0,0}+u_{1,0}) \right)= \alpha^{-1} \partial_{u_{2,0}} q_{0,0},\\
		r_{0,0} &=& u_{2,0} u_{1,0} (1+ \alpha u_{-1,0})- u_{-2,0} u_{-1,0} (1+ \alpha u_{1,0})+ \alpha u_{2,0} u_{-2,0} (u_{1,0}-u_{-1,0}).
		\end{eqnarray}
	\end{subequations}
It should be noted that the Miura transformation $w_{0,0} = u_{2,0} p_{0,0}/q_{0,0}$ maps symmetry (\ref{eq:sym-V}) to the Bogoyavlensky lattice $\partial_t w_{0,0} = w_{0,0} (1+ a w_{0,0}) (w_{3,0} w_{2,0} w_{1,0} - w_{-1,0} w_{-2,0} w_{-3,0})$. \hfill $\Box$
\end{example}

\section{Overdetermined systems of difference equations and consistency} \label{sec:def}

The systems we discussed in the previous section have three properties in common.
\begin{enumerate}
	\item The two equations constituting these systems are defined on different stencils. The first equation of these systems is defined on two consecutive quadrilaterals in the vertical direction, whereas the second equation is given on two consecutive quadrilaterals horizontally. The two stencils form a staircase with two steps and their intersection is an elementary quadrilateral on the lattice. 
	\item Every equation of the system can be solved uniquely for the values of $u$ at the corners of the rectangular stencil they are defined.  Specifically, equations  (\ref{eq:nonint-sys-E1}), (\ref{eq:tau-H1-E1}), (\ref{eq:intro-sys-E1}) and (\ref{eq:ex2-e1}) can be solved uniquely for $u_{0,0}$, $u_{0,2}$, $u_{1,0}$ and $u_{1,2}$, whereas  (\ref{eq:nonint-sys-E2}), (\ref{eq:tau-H1-E2}), (\ref{eq:intro-sys-E2}) and (\ref{eq:ex2-e2}) for $u_{0,0}$, $u_{2,0}$, $u_{0,1}$ and $u_{2,1}$.
	\item They are  consistent.
\end{enumerate}
Using these properties as a prototype, we propose their generalization to overdetermined systems involving $N$ equations for one function $u$. More precisely, we consider overdetermined systems of $N$ equations for a scalar function $u$ which satisfy the following three properties. For simplicity in what follows we denote such a system with $C_N$ and refer to $N$ as its order. 
\begin{enumerate}
\item[R1.] {\emph{Each equation of the system is defined on a different stencil.}} \\ More precisely, with a given integer $N$ we consider the line $n+m=N+1$ on the ${\mathbb{Z}}^2$ lattice and  the right isosceles triangle $\Delta_N$ with vertices at the points $(0,0)$, $(N+1,0)$ and $(0,N+1)$. The $N$ rectangles ${\cal{R}}_i$ inscribed in $\Delta_N$ with vertices at the lattice points $(0,0)$, $(i,0)$, $(0,j)$ and $(i,j)$, with $i+j=N+1$ and  $i=1,\ldots, N$, are the stencils of the $N$ equations of the system, i.e.
\begin{equation} \label{eq:gen-cn}
C_N = \left\{ E_{i}\left(u_{0,0},\ldots,u_{i,0}, \ldots, u_{0,j},\ldots,u_{i,j}\right) = 0,\quad  i=1,\ldots,N \,\,, {\mbox{ and }}\,\,\, j= N-i+1\right\} .
\end{equation}

\begin{figure}[!htb]
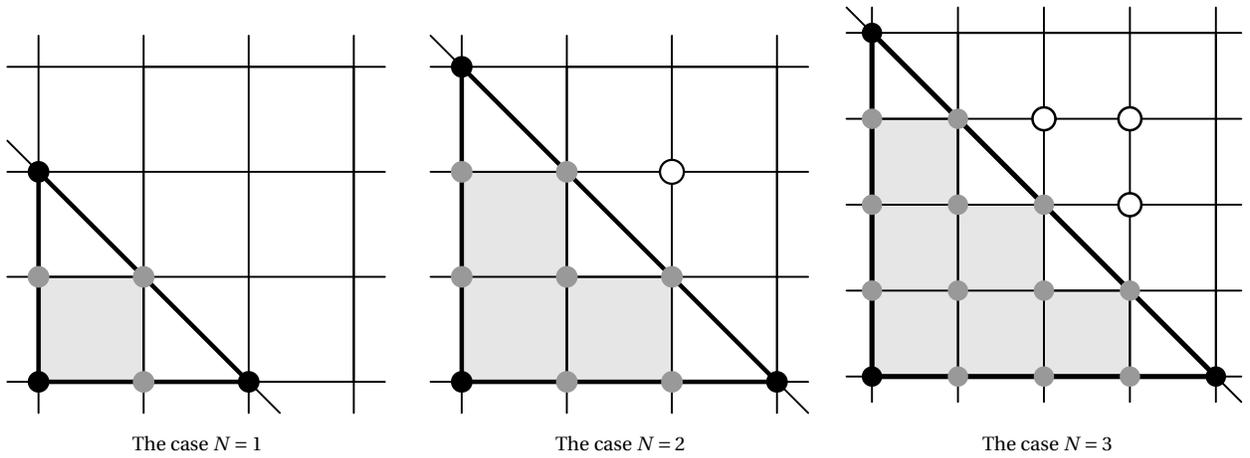

	\begin{minipage}[b]{12pc}
		\centering
		\centertexdraw{ \setunitscale 0.55
			\linewd 0.02 \arrowheadtype t:F 
			\htext(0 0.5) {\phantom{T}}
			\move (0 0) \lvec (1 0) \lvec (1 1) \lvec (0 1) \lvec (0 0) \lfill f:0.9
			\move (1 0) \lvec (3 0) \lvec (3 3) \lvec(1 3) \lvec(1 0) 
			\move (-0.3 0) \lvec(3.3 0)
			\move (-0.3 1) \lvec(3.3 1)
			\move (-0.3 2) \lvec(3.3 2)
			\move (-0.3 3) \lvec(3.3 3)
			\move (0 -0.3) \lvec(0 3.3)
			\move (1 -0.3) \lvec(1 3.3)
			\move (2 -0.3) \lvec(2 3.3)
			\move (3 -0.3) \lvec(3 3.3)
			\move(-0.3 2.3) \lvec (2.3 -0.3)
			\linewd 0.04 \move (0 0) \lvec (2 0) \lvec (0 2) \lvec(0 0)
			\move (0 0) \fcir f:0 r:0.1 \move (1 0) \fcir f:0.6 r:0.1
			\move (2 0) \fcir f:0 r:0.1 
			\move (0 1) \fcir f:0.6 r:0.1 \move (1 1) \fcir f:0.6 r:0.1
			\move (0 2) \fcir f:0 r:0.1
			\htext (0.9 -.65) {\footnotesize{The case $N=1$}} 			
		}
	\end{minipage}
	\hspace{1.pc}
	\begin{minipage}[b]{12pc}
		\centering
		\centertexdraw{ \setunitscale 0.55
			\linewd 0.02 \arrowheadtype t:F 
			\htext(0 0.5) {\phantom{T}}
			\move (0 0) \lvec (2 0) \lvec (2 1) \lvec (0 1) \lvec (0 0) \lfill f:0.9
			\move (0 0) \lvec (1 0) \lvec (1 2) \lvec (0 2) \lvec (0 0) \lfill f:0.9
			\move (1 0) \lvec (3 0) \lvec (3 3) \lvec(1 3) \lvec(1 0) 
			\move (-0.3 0) \lvec(3.3 0)
			\move (-0.3 1) \lvec(3.3 1)
			\move (-0.3 2) \lvec(3.3 2)
			\move (-0.3 3) \lvec(3.3 3)
			\move (0 -0.3) \lvec(0 3.3)
			\move (1 -0.3) \lvec(1 3.3)
			\move (2 -0.3) \lvec(2 3.3)
			\move (3 -0.3) \lvec(3 3.3)
			\move(-0.3 3.3) \lvec (3.3 -0.3)
			\linewd 0.04 \move (0 0) \lvec (3 0) \lvec (0 3) \lvec(0 0)
			%\lpatt (.05 .05) \move (0 2) \lvec(2 0) \lvec( 0 0) \lvec(0 2) \lpatt()
			\move (0 0) \fcir f:0 r:0.1 \move (1 0) \fcir f:0.6 r:0.1
			\move (2 0) \fcir f:0.6 r:0.1 \move (3 0) \fcir f:0 r:0.1
			\move (0 1) \fcir f:0.6 r:0.1 \move (1 1) \fcir f:0.6 r:0.1
			\move (2 1) \fcir f:0.6 r:0.1 
			\move (0 2) \fcir f:0.6 r:0.1 \move (1 2) \fcir f:0.6 r:0.1
			\move (0 3) \fcir f:0 r:0.1
			\move (2 2) \lcir r:0.1 \fcir f:1 r:0.1
			\htext (0.9 -.65) {\footnotesize{The case $N=2$}} 			
		}
	\end{minipage}
	\hspace{1.pc}
	\begin{minipage}[b]{12pc}
		\centertexdraw{ \setunitscale 0.45
			\linewd 0.02 \arrowheadtype t:F 
			\htext(0 0.5) {\phantom{T}}
			\move (0 0) \lvec (3 0) \lvec (3 1) \lvec (0 1) \lvec (0 0) \lfill f:0.9
			\move (0 0) \lvec (1 0) \lvec (1 3) \lvec (0 3) \lvec (0 0) \lfill f:0.9				
			\move (0 0) \lvec (2 0) \lvec (2 2) \lvec (0 2) \lvec (0 0) \lfill f:0.9
			\move (1 0) \lvec (4 0) \lvec (4 4) \lvec(1 4) \lvec(1 0) 
			\move (-0.3 0) \lvec(4.3 0)
			\move (-0.3 1) \lvec(4.3 1)
			\move (-0.3 2) \lvec(4.3 2)
			\move (-0.3 3) \lvec(4.3 3)
			\move (-0.3 4) \lvec(4.3 4)
			\move (0 -0.3) \lvec(0 4.3)
			\move (1 -0.3) \lvec(1 4.3)
			\move (2 -0.3) \lvec(2 4.3)
			\move (3 -0.3) \lvec(3 4.3)
			\move (4 -0.3) \lvec(4 4.3)				
			\move(-0.3 4.3) \lvec (4.3 -0.3)
			%\lpatt (.05 .05) \move (0 3) \lvec(3 0) \lvec(0 0) \lvec(0 3) \lpatt()
			\linewd 0.06 \move (0 0) \lvec (4 0) \lvec (0 4) \lvec(0 0)
			\move (0 0) \fcir f:0 r:0.12 \move (1 0) \fcir f:0.6 r:0.12
			\move (2 0) \fcir f:0.6 r:0.12 \move (3 0) \fcir f:0.6 r:0.12
			\move (0 1) \fcir f:0.6 r:0.12 \move (1 1) \fcir f:0.6 r:0.12
			\move (2 1) \fcir f:0.6 r:0.12 \move (3 1) \fcir f:0.6 r:0.12
			\move (0 2) \fcir f:0.6 r:0.12 \move (1 2) \fcir f:0.6 r:0.12 \move (2 2) \fcir f:0.6 r:0.12
			\move (0 3) \fcir f:0.6 r:0.12 \move (1 3) \fcir f:0.6 r:0.12
			\move (0 4) \fcir f:0 r:0.12 \move (4 0) \fcir f:0 r:0.12		
			\move (3 2) \lcir r:0.12 \fcir f:1 r:0.12 
			\move (2 3) \lcir r:0.12 \fcir f:1 r:0.12
			\move (3 3) \lcir r:0.12 \fcir f:1 r:0.12									
			\htext (1.3 -.85) {\footnotesize{The case $N=3$}}
		}
	\end{minipage}
\caption{The stencils of the equations for $C_1$, $C_2$ and $C_3$.} \label{fig:cn}
\end{figure}

\item[R2.] {\emph{Each equation of the system can be solved uniquely for any of the values of $u$ at the corners of the rectangle it is defined}}. \\
This means that $E_i=0$ can be solved uniquely for any of $u_{0,0}$, $u_{i,0}$, $u_{0,N+1-i}$ and $u_{i,N+1-i}$. \\ 
A consequence of this requirement is that system $C_N$ can be solved uniquely for any set of values of $u$ lying on the same edge of the triangle  $\Delta_N$.
\item[R3.] {\emph{System $C_N$ is consistent.}}\\
The previous requirement along with the fact that variable $u_{i,N+1-i}$ appears only in $E_i =0$ imply that $C_N$ can always be solved uniquely for $(u_{1,N},\ldots,u_{N,1})$. In particular this allows us to rewrite system (\ref{eq:gen-cn}) in the  solved form
\begin{equation}\label{eq:gen-cn-1}
C_N = \left\{u_{i,N+1-i} = F_{i}\left(u_{0,0},\ldots,u_{i,0}, \ldots, u_{0,N+1-i},\ldots,u_{i-1,N+1-i}\right),\quad i=1,\ldots,N\right\}.
\end{equation}
Using this equivalent form of $C_N$, we define consistency as follows.
\begin{definition} \label{def:con}
We call system (\ref{eq:gen-cn-1}) consistent if the following relations hold on solutions of system (\ref{eq:gen-cn-1}).
\begin{equation} \label{eq:con-con}
 {\cal{T}}^{i-j}(F_{i}) - {\cal{S}}^{i-j}(F_j) = 0,\quad \forall \,\, i>j.
 \end{equation}
\end{definition}

\begin{remark} \label{rem:con-ch}
It is sufficient to check only the consistency of consecutive equations, i.e. conditions (\ref{eq:con-con}) with $(i,j) =(\ell+1,\ell)$, for all $\ell=1,\ldots,N-1$. 
\end{remark}

Alternatively, we can state that the system is consistent if the values $u_{i,j}$, with $0 < i,j \le N$ and $N+1<i+j \le 2 N$, can be found uniquely using the equations of $C_N$.  For example when $N=2$ and $3$ this means to find uniquely the values of $u$ at the white disks in Figure \ref{fig:cn}. It should be noted that all these values are in general functions of the $\tfrac{(N+1)(N+2)}{2}$ values of $u$ involved in $\Delta_{N-1}$.
\end{enumerate}

\begin{remark} \label{rem:con-com}
The case $N=1$ corresponds to scalar quad equations, see also Figure \ref{fig:cn}, for which obviously the above notion of consistency is not applicable. However we include quad equations in our considerations because they may be interpreted as the first members of hierarchies of integrable consistent systems, see also next section. When $N=2$ the three requirements R1--R3 clearly coincide with the properties we listed at the beginning of this section. \hfill $\Box$
\end{remark}

\begin{example}
From the previous remark it is obvious that the second order systems  (\ref{eq:nonint-sys}), (\ref{eq:tau-H1}), (\ref{eq:intro-sys}) and (\ref{eq:ex2-e1}, \ref{eq:ex2-e2}) satisfy the three requirements R1-R3. It is not difficult to see that the third order system 
\begin{subequations}\label{eq:sys-ord3}
\begin{eqnarray}
&& u_{0,3} (u_{0,0}-u_{1,1}) (u_{0,1}-u_{1,2})(u_{0,2}-u_{1,3}) + \nonumber \\ 
&&\qquad \qquad u_{1,0} \left( u_{0,2} (u_{0,0}-u_{1,1}) (u_{0,1}-u_{1,2}) + u_{1,3} (u_{0,0} (u_{1,2}-u_{0,1})+u_{0,1} u_{1,1})\right) =0,\\
&& u_{0,2} (u_{0,0}-u_{1,1}) (u_{0,1}-u_{1,2}) (u_{1,0}-u_{2,1}) (u_{1,1}-u_{2,2}) + \nonumber \\
&& \qquad \qquad u_{2,0} \left( u_{0,1} u_{1,1} u_{2,1} (u_{1,1}-u_{0,0}) + u_{1,0} (u_{1,1}-u_{2,2}) (u_{0,0} (u_{0,1}-u_{1,2})-u_{0,1} u_{1,1}  \right)=0,\\
&& u_{0,1} (u_{0,0}-u_{1,1}) (u_{1,0}-u_{2,1}) (u_{2,0}-u_{3,1}) + u_{0,0} u_{1,0} u_{2,0} u_{3,0} =0,
\end{eqnarray}
\end{subequations}
satisfies R1 and R2. For the consistency requirement we  write the system as 
\begin{eqnarray*}
&& u_{1,3} =\frac{u_{0,2} \left(u_{0,3}+u_{1,0}\right) \left(u_{0,0}-u_{1,1}\right) \left(u_{0,1}-u_{1,2}\right)}{u_{0,0} \left(u_{0,3}+u_{1,0}\right) \left(u_{0,1}-u_{1,2}\right)+u_{1,1} \left(u_{0,3} u_{1,2}-u_{0,1} \left(u_{0,3}+u_{1,0}\right)\right)},\\ 
&& u_{2,2} = \frac{u_{1,1} \left(u_{1,0} \left(u_{0,0} \left(u_{0,1}-u_{1,2}\right)-u_{0,1} u_{1,1}\right) u_{2,0}+u_{0,1} \left(u_{1,1}-u_{0,0}\right) u_{2,1} u_{2,0}+u_{0,2} \left(u_{0,0}-u_{1,1}\right) \left(u_{0,1}-u_{1,2}\right) \left(u_{1,0}-u_{2,1}\right)\right)}{u_{1,0} \left(u_{0,0} \left(u_{0,1}-u_{1,2}\right)-u_{0,1} u_{1,1}\right) u_{2,0}+u_{0,2} \left(u_{0,0}-u_{1,1}\right) \left(u_{0,1}-u_{1,2}\right) \left(u_{1,0}-u_{2,1}\right)},\\ 
&& u_{3,1} =  u_{2,0}+\frac{u_{0,0} u_{1,0} u_{2,0} u_{3,0}}{u_{0,1} \left(u_{0,0}-u_{1,1}\right) \left(u_{1,0}-u_{2,1}\right)},
\end{eqnarray*}
and then check if the compatibility conditions ${\cal{S}}(u_{1,3}) = {\cal{T}}(u_{2,2})$, ${\cal{S}}(u_{2,2}) = {\cal{T}}(u_{3,1})$, ${\cal{S}}^2(u_{1,3}) = {\cal{T}}^2(u_{3,1})$ hold on solutions of the system. For the first two conditions we have to take into account only the system, whereas for the last one we have to use also the shifts of the system in order to replace $u_{2,3}$ and $u_{3,2}$. After some calculations with the help of symbolic software it follows that these conditions do hold on solutions of (\ref{eq:sys-ord3}) and thus the system is consistent.  \hfill $\Box$
\end{example}

Our requirements for the solvability of $C_N$ allow us to determine uniquely the solution of the system once appropriate initial values are given. More precisely,

\begin{proposition}
Consider the infinitely extended edges of triangle $\Delta_N$, i.e. the lines $n=0$, $m=0$ and $n+m= N+1$.
If initial values are given at 
\begin{enumerate}
\item  all the points on any two of these three lines, i.e. any two of the sets of values $\{u_{0,k}\}$, $\{u_{k,0}\}$ and $\{u_{k,N-k+1}\}$ for all $k\in {\mathbb{Z}}$,
\item  and all the interior points of $\Delta_N$, i.e. $\{u_{a,b}\}$, for all $0<a,b<N$ with $a+b <N+1$, 
\end{enumerate}
then the solution $u$ of the consistent system $C_N$ can be determined uniquely everywhere on the ${\mathbb{Z}}^2$ lattice.
\end{proposition}

In particular, we refer to the values of $u$ along the lines $m=0$ and $n=0$, i.e. $u_{k,0}$ and $u_{0,k}$ for all $k \in {\mathbb{Z}}$, and all the interior points of $\Delta_N$ as  {\emph{ the standard dynamical variables}}.

\begin{center}
	\begin{figure}[th]
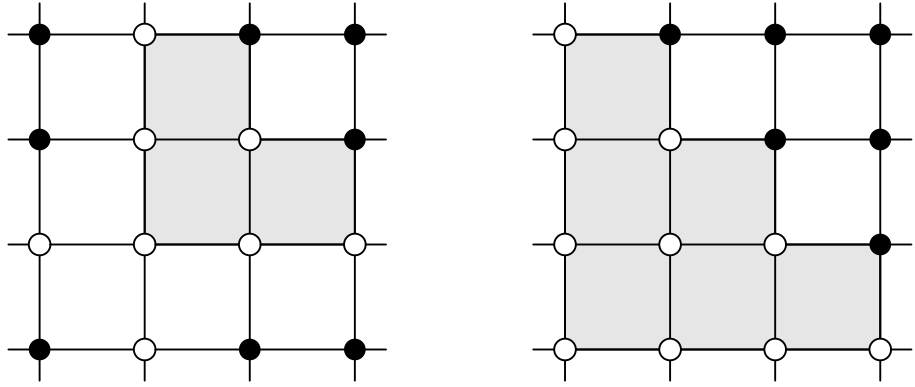

		\centertexdraw{ \setunitscale 0.55
			\linewd 0.02 \arrowheadtype t:F 
			\htext(0 0.5) {\phantom{T}}
			\move (1 1) \lvec (3 1) \lvec (3 2) \lvec(2 2) \lvec (2 3) \lvec (1 3) \lvec(1 1) \lfill f:0.9
			\move (-0.3 0) \lvec(3.3 0)
			\move (-0.3 1) \lvec(3.3 1)
			\move (-0.3 2) \lvec(3.3 2)
			\move (-0.3 3) \lvec(3.3 3)
			\move (0 -0.3) \lvec(0 3.3)
			\move (1 -0.3) \lvec(1 3.3)
			\move (2 -0.3) \lvec(2 3.3)
			\move (3 -0.3) \lvec(3 3.3)
			\move (1 1) \fcir f:1 r:0.1 \move(1 1) \lcir r:0.1 
			\move (2 1) \fcir f:1 r:0.1 \move(2 1) \lcir r:0.1 
			\move (3 1) \fcir f:1 r:0.1 \move(3 1) \lcir r:0.1 
			\move (1 2) \fcir f:1 r:0.1 \move(1 2) \lcir r:0.1 
			\move (2 2) \fcir f:1 r:0.1 \move(2 2) \lcir r:0.1 
			\move (1 3) \fcir f:1 r:0.1 \move(1 3) \lcir r:0.1  
			\move (0 1) \fcir f:1 r:0.1 \move(0 1) \lcir r:0.1 
			\move (1 0) \fcir f:1 r:0.1 \move(1 0) \lcir r:0.1
			\move (0 0) \fcir f:0 r:0.1 \move(0 2) \fcir f:0 r:0.1
			\move (0 3) \fcir f:0 r:0.1 \move (2 0) \fcir f:0 r:0.1
			\move (2 3) \fcir f:0 r:0.1 \move (3 0) \fcir f:0 r:0.1
			\move (3 2) \fcir f:0 r:0.1 \move (3 3) \fcir f:0 r:0.1
			
			\move (5 0) \lvec (8 0) \lvec (8 1) \lvec(7 1) \lvec(7 2) \lvec (6 2) \lvec (6 3) \lvec (5 3) \lvec (5 0) \lfill f:0.9
			\move (4.7 0) \lvec(8.3 0)
			\move (4.7 1) \lvec(8.3 1)
			\move (4.7 2) \lvec(8.3 2)
			\move (4.7 3) \lvec(8.3 3)
			\move (5 -0.3) \lvec(5 3.3)
			\move (6 -0.3) \lvec(6 3.3)
			\move (7 -0.3) \lvec(7 3.3)
			\move (8 -0.3) \lvec(8 3.3)
			\move (5 0) \fcir f:1 r:0.1 \move(5 0) \lcir r:0.1 
			\move (6 0) \fcir f:1 r:0.1 \move(6 0) \lcir r:0.1 
			\move (7 0) \fcir f:1 r:0.1 \move(7 0) \lcir r:0.1 
			\move (8 0) \fcir f:1 r:0.1 \move(8 0) \lcir r:0.1 
			\move (5 1) \fcir f:1 r:0.1 \move(5 1) \lcir r:0.1 
			\move (5 2) \fcir f:1 r:0.1 \move(5 2) \lcir r:0.1 
			\move (5 3) \fcir f:1 r:0.1 \move(5 3) \lcir r:0.1 
			\move (6 1) \fcir f:1 r:0.1 \move(6 1) \lcir r:0.1 
			\move (6 2) \fcir f:1 r:0.1 \move(6 2) \lcir r:0.1 
			\move (7 1) \fcir f:1 r:0.1 \move(7 1) \lcir r:0.1
			\move (8 1) \fcir f:0 r:0.1 \move (7 2) \fcir f:0 r:0.1
			\move (8 2) \fcir f:0 r:0.1 \move (6 3) \fcir f:0 r:0.1
			\move (7 3) \fcir f:0 r:0.1 \move (8 3) \fcir f:0 r:0.1
		}
		\caption{\emph{\small{Standard dynamical variables (white disks) for $N=2$ and $N=3$. If initial values are given at the white vertices, then solution $u$ can be found uniquely at any other lattice point (black disks).}}} \label{fig:ivp}
	\end{figure}
\end{center}

\noindent The standard dynamical variables are of particular interest as they are involved in the generalized symmetries and the integrability of the underlying consistent system.
\begin{definition}\label{def:int-1}
We call consistent system $C_N$ integrable if it admits infinite hierarchies of symmetries which depend on a finite but otherwise unspecified number of standard dynamical variables.
\end{definition}
All the consistent systems we have at our disposal admit two hierarchies of symmetries none of which involve any dynamical  variable $u_{a,b}$ with $0<a,b<N$ and $a+b <N+1$. Thus we can  slightly modify the above definition as follows.
\begin{definition}\label{def:int}
We call consistent system $C_N$ integrable if it admits infinite hierarchies of symmetries in both lattice directions each one of which depends on a finite but otherwise unspecified number of dynamical variables $u_{k,0}$ or $u_{0,k}$ only.
\end{definition}

\begin{example}
The second order systems (\ref{eq:intro-sys}) and (\ref{eq:ex2-e1}, \ref{eq:ex2-e2}) are integrable and their lowest order symmetries were given in the previous section in Examples \ref{ex:deg} and \ref{ex:const}, respectively. The third order system (\ref{eq:sys-ord3}) is also integrable and its lowest order symmetries are generated by
\begin{equation}\label{eq:sys-ord3-sm1}
\partial_t u_{0,0} = u_{0,0} \left(\frac{u_{3,0}}{u_{-1,0}} + \frac{u_{2,0}}{u_{-2,0}} + \frac{u_{1,0}}{u_{-3,0}}\right)
\end{equation}
and
\begin{equation}\label{eq:sys-ord3-sm2}
\partial_s u_{0,0} = \frac{u_{0,0} u_{0,1} u_{0,2} u_{0,3}}{(u_{0,3}+u_{0,-1}) (u_{0,2}+u_{0,-2}) (u_{0,1}+u_{0,-3}) },
\end{equation}
respectively. \hfill $\Box$
\end{example}

\section{Lattice paths and consistent systems of difference equations} \label{sec:sys}

Having developed a general framework for consistent systems, in this section we present the construction of a hierarchy of consistent systems which employs lattice paths. We discuss the properties of these systems and prove their integrability by deriving their symmetries. Moreover we present a deformation for the first three members of this family and discuss their relations to known quad equations.

We start our derivations with the construction of certain polynomials which will be the building blocks of the hierarchy of consistent systems.
\begin{itemize}
\item[1.] Consider all the lattice paths from $(0,0)$ to $(i,j)$, where $i \ge 0$, $j \ge 0$ and $i+j > 0$, which can be constructed by moving only parallel to the positive direction of either axis. For  every choice of $i$ and $j$ there exist $\tfrac{(i+j)!}{i! j!}$ different paths which connect $i+j+1$ points on the lattice, including the origin and the endpoint $(i,j)$. We denote these paths with ${\cal{P}}_{(i,j)}^{(a)}$, where $a = 1,\ldots,\tfrac{(i+j)!}{i! j!}$.

\item[2.] With every path ${\cal{P}}_{(i,j)}^{(a)}$ we associate the product of the values of the function $u$ at the $i+j+1$ lattice points connected by the path, 
$${\cal{P}}_{(i,j)}^{(a)} = u_{0,0} \overset{i+j-1\, \text{ terms}}{\cdots \cdots \cdots} u_{i,j}.$$

\item[3.] With the above association, we define the {\emph{multilinear and homogeneous polynomials of degree}} $i+j+1$
\begin{equation}
Q_{(i,j)} = \sum_{a=1}^{(i+j)!/i! j!} {\cal{P}}^{(a)}_{(i,j)}, \quad {\mbox{with }} \quad i \ge 0,\,\,j \ge 0 {\mbox{~ and ~}} i+j > 0.
\end{equation}

By exploiting the combinatorics in the construction of polynomials $Q_{(i,j)}$, we can find two different ways to determine these polynomials recursively as it is described below.

\begin{lemma} \label{lem:def-Q}
If we define
\begin{subequations} \label{eq:def-Q}
\begin{equation}
\begin{array}{ll} Q_{(i,j)} = 0, &{\text{ if at least one index is negative}},  \\ Q_{(0,0)} =  u_{0,0}, & \end{array}
\end{equation}
then polynomials $Q_{(i,j)}$, with $i,j \ge 0$ and $i+j>0$, can be determined recursively by
\begin{equation} \label{eq:def-Qb}
Q_{(i,j)} = u_{0,0} \left\{{\cal{S}} \left(Q_{(i-1,j)}\right) + {\cal{T}} \left(Q_{(i,j-1)}\right)\right\} ,
\end{equation}
or
\begin{equation} \label{eq:def-Qc}
Q_{(i,j)} = \left\{Q_{(i,j-1)} + Q_{(i-1,j)}\right\}  u_{i,j} . 
\end{equation}
\end{subequations}
\end{lemma}

\begin{proof}
Since we start always from the origin and we can make only one step every time either right or up, initially we can move from $u_{0,0}$ either to $u_{1,0}$ or to $u_{0,1}$, respectively. Then we use the paths starting from $(1,0)$ terminating at $(i,j)$ which are encoded into ${\cal{S}} (Q_{(i-1,j)})$, and the ones from $(0,1)$ ending at $(i,j)$ given by ${\cal{T}}(Q_{(i,j-1)})$. This observation and the properties of the polynomials lead to the first recursive definition (\ref{eq:def-Qb}). Alternatively, we can reach point $(i,j)$ either from $(i,j-1)$ by moving one step up, or from $(i-1,j)$ by making one step right. The first approach is equivalent to $Q_{(i,j-1)} u_{i,j}$ and the second one to $Q_{(i-1,j)} u_{i,j}$ whereas their sum gives the second definition (\ref{eq:def-Qc}).
\end{proof}
\end{itemize}

With the above polynomials at our disposal and for any $N \ge 1$, we define the overdetermined system of equations
\begin{equation} \label{eq:SN}
\Sigma_N = \left\{Q_{(i,N-i+1)}  +  (-1)^{N-i} \alpha_N =0 \, , \quad i = 1, \ldots, N  \right\},
\end{equation}
where $\alpha_N \in {\mathbb{R}}^*$ is a parameter.

The geometric construction of $Q_{(i,j)}$ and their properties clearly imply that system $\Sigma_N$ satisfies requirements R1 and R2. Moreover, 
\begin{proposition}\label{prop:con}
System $\Sigma_N$ is consistent.
\end{proposition}

\begin{proof}
To check the consistency of system $\Sigma_{N}$ (\ref{eq:SN}) first we solve its equations for $u_{i,N-i+1}$. In view of (\ref{eq:def-Qc}) this leads to
\begin{equation} \label{eq:con-u}
 u_{i,N-i+1} = \frac{-(-1)^{N-i} \alpha_{N}}{ F_i} := \frac{-(-1)^{N-i} \alpha_{N}}{Q_{(i-1,N-i+1)} + Q_{(i,N-i)}},\quad i=1,\ldots N.
 \end{equation}
Next we have to examine if $ (-1)^i {\cal{T}}^{i-j}(F_i) = (-1)^j {\cal{S}}^{i-j}(F_j)$, on solutions of $\Sigma_N$ for all $i,j=1,\ldots,N$ and $i > j$. For our purposes it is sufficient to see if these relations hold for any pair of consecutive values for indices $i$ and $j$, i.e. for any $(i,j) = (\ell+1,\ell)$ with $\ell=1, \ldots, N-1$. With these choices the above requirements become
\begin{eqnarray*}
{\cal{S}}(F_{\ell}) + {\cal{T}}(F_{\ell+1}) &=& {\cal{S}}\Big(Q_{(\ell-1,N-\ell+1)} + Q_{(\ell,N-\ell)}\Big) + {\cal{T}}\Big(Q_{(\ell,N-\ell)} + Q_{(\ell+1,N-\ell-1)}\Big)  \\
	&=&  {\cal{S}}(Q_{(\ell-1,N-\ell+1)}) + {\cal{T}} (Q_{(\ell,N-\ell)}) +  {\cal{S}}(Q_{(\ell,N-\ell)}) + {\cal{T}} (Q_{(\ell+1,N-\ell-1)}) \\
	&=& \frac{Q_{(\ell,N-\ell+1)}}{u_{0,0}} + \frac{Q_{(\ell+1,N-\ell)}}{u_{0,0}} = \frac{-(-1)^{N-\ell} \alpha_N}{u_{0,0}} + \frac{-(-1)^{N-\ell-1} \alpha_N}{u_{0,0}}  = 0,
\end{eqnarray*}
where we have also used (\ref{eq:def-Qb}) and (\ref{eq:SN})  in the last two steps, respectively. This clearly shows that for any two consecutive  values of $i$, relations (\ref{eq:con-u}) are consistent on solutions of $\Sigma_{N}$, and thus $\Sigma_{N}$ is consistent. 
\end{proof}

Hence $\Sigma_N$ satisfies all three requirements R1--R3. To prove $\Sigma_{N}$ is integrable, we study the symmetries of two equations, namely of $Q_{(N,1)} + \alpha_N =0$ and $Q_{(1,N)} - (-1)^N \alpha_{N}=0$, using the method of \cite{X3}.

\begin{proposition} \label{prop:sym}
Equation $Q_{(N,1)} + \alpha_N = 0$ admits infinite hierarchies of generalised symmetries in the first direction. The first member of this hierarchy has order $N+1$ and is generated by
\begin{equation}\label{eq:UN} 
\partial_t u_{0,0} \,=\, u_{0,0}  \left({\cal{S}}-1\right) \prod_{k=0}^{N} {\cal{S}}^{k-N-1} \left(\frac{1}{Q_{(N+1,0)} \,-\,\alpha_N} \right).
\end{equation}
Respectively, equation $Q_{(1,N)} - (-1)^N \alpha_N = 0$ admits infinite hierarchies of generalised symmetries in the second direction. The first member of this hierarchy has order $N+1$ and is generated by
\begin{equation}\label{eq:UM}
\partial_s u_{0,0} \,=\, u_{0,0}  \left({\cal{T}}-1\right) \prod_{k=0}^{N} {\cal{T}}^{k-N-1} \left(\frac{1}{Q_{(0,N+1)} \,+\,(-1)^N \alpha_N} \right).
\end{equation}
\end{proposition}

\begin{proof}
Because of the invariance of $\Sigma_N$ under the transformation $\left(u_{k,l},\alpha_N \right) \mapsto \left(u_{l,k},(-1)^{N+1} \alpha_N \right)$, it is sufficient to study the symmetries of equation $Q_{(N,1)}+\alpha_N=0$. This can be done on a case-by-case basis and it is sufficient to show that $\partial_t Q_{(N,1)} =0$ on solutions of $Q_{(N,1)}+\alpha_N=0$ (see also \cite{MX} for $N=1$ and \cite{X3} for $N=2$).
\end{proof}

We can now extend the symmetries of these equations to symmetries of system $\Sigma_{N}$.
\begin{corollary}
The differential-difference equations (\ref{eq:UN}) and (\ref{eq:UM}) define the lowest order symmetries of system $\Sigma_{N}$.
\end{corollary}

\begin{proof}
Firstly we observe that relations
\begin{equation} \label{eq:con2}
{\cal{T}}^p(Q_{(N,1)}) = (-1)^p {\cal{S}}^p (Q_{(N-p,p+1)}), \quad {\cal{S}}^p(Q_{(1,N)}) =  (-1)^p {\cal{T}}^p (Q_{(p+1,N-p)}),\quad p= 1,\ldots, N-1,
\end{equation}
hold on solutions of $\Sigma_{N}$ as a consequence of the consistency of $\Sigma_N$. It follows from the first relation in (\ref{eq:con2}) that $Q_{(N-p,p+1)}= (-1)^p {\cal{S}}^{-p} {\cal{T}}^p(Q_{(N,1)})$ for all $p=1,\ldots,N-1$, and thus  
$$\partial_t Q_{(N-p,p+1)}= (-1)^p {\cal{S}}^{-p} {\cal{T}}^p( \partial_t Q_{(N,1)}).$$
But since $\partial_t Q_{(N,1)} =0$ on solutions of $\Sigma_{N}$, we conclude that also $\partial_t Q_{(N-p,p+1)}=0$. Similarly the second relation in (\ref{eq:con2}) leads to $Q_{(p+1,N-p)}= (-1)^p {\cal{S}}^{p} {\cal{T}}^{-p}(Q_{(1,N)})$ and subsequently to $\partial_s Q_{(p+1,N-p)}= (-1)^p {\cal{S}}^{p} {\cal{T}}^{-p}(\partial_sQ_{(1,N)})$. Since $\partial_sQ_{(1,N)}=0$ on solutions of $\Sigma_{N}$, we arrive at $\partial_s Q_{(p+1,N-p)}= 0$.
\end{proof}

\begin{remark}A final remark is that the difference substitution
\begin{equation} \label{eq:difsub}
v_{0,0} = \frac{1}{Q_{(N+1,0)} - \alpha_{N}}
\end{equation}
maps (\ref{eq:UN}) to the Bogoyavlensky lattice
\begin{equation}\label{eq:Bog}
\partial_t v_{0,0} = - v_{0,0} (\alpha_{N} v_{0,0} + 1) \left( v_{N+1,0} \ldots v_{1,0} - v_{-1,0} v_{-2,0} \ldots v_{-N-1,0} \right).
\end{equation}
Indeed, in terms of the above substitution symmetry (\ref{eq:UN}) can be written as $\partial_t u_{0,0} = u_{0,0} ({\cal{S}}-1) \prod_{k=0}^{N-1} v_{k-N-1,0}$, whereas the $t$-derivative of (\ref{eq:difsub}) is
 $$\partial_t v_{0,0} = -v_{0,0}^2 Q_{(N+1,0)} \sum_{i=0}^{N+1}  \frac{\partial_t u_{i,0}}{u_{i,0}}=  -v_{0,0} (\alpha_N v_{0,0}+1)  \sum_{i=0}^N  ({\cal{S}}-1) \prod_{k=0}^{N+1} v_{i+k-N-1}=  -v_{0,0} (\alpha_N v_{0,0}+1) \left( \prod_{i=1}^{N+1} v_{i,0} - \prod_{i=1}^{N+1} v_{-i,0}  \right).$$ 
 Similar considerations clearly hold for (\ref{eq:UM}). \hfill $\Box$
\end{remark}

We can easily implement recursive formulae (\ref{eq:def-Q}) for the construction of $\Sigma_N$ and below  we give the systems which correspond to $N=1$, $2$ and $3$.
\begin{itemize}
\item[1.] System $\Sigma_1$ is the known quadrilateral equation 
\begin{equation} \label{eq:2}
u_{0,0} \left(u_{1,0}  + u_{0,1}\right)  u_{1,1} + \alpha_1 = 0,
\end{equation}
which was first given in \cite{MX} along with its lowest order symmetries. 
\begin{equation} \label{eq:2sym}
\partial_t u_{0,0} = \frac{u_{0,0}^2 (u_{2,0} u_{1,0} - u_{-1,0} u_{-2,0})}{\prod_{i=0}^2 (u_{i,0} u_{i-1,0} u_{i-2,0} - \alpha_1 )},\quad 
\partial_s u_{0,0} = \frac{u_{0,0}^2 (u_{0,2} u_{0,1} - u_{0,-1} u_{0,-2})}{\prod_{i=0}^2 (u_{0,i} u_{0,i-1} u_{0,i-2} - \alpha_1)}.
\end{equation}
It was also derived in a different context in \cite{FX}. 

\item[2.] System $\Sigma_2$ is constituted by the two equations
\begin{subequations}\label{eq:3}
\begin{eqnarray}
&& u_{0,0}  \left( u_{1,0} u_{1,1} + u_{0,1} u_{1,1} +u_{0,1} u_{0,2} \right) u_{1,2} - \alpha_2 = 0,\\
&&  u_{0,0} \left( u_{1,0} u_{2,0} + u_{1,0} u_{1,1}  +u_{0,1} u_{1,1}\right) u_{2,1} + \alpha_2   = 0.
\end{eqnarray}
\end{subequations}
Its lowest order symmetries in both directions are generated by 
\begin{equation}\label{eq:3sym}
\partial_t u_{0,0} = \frac{u_{0,0}^2 (u_{3,0} u_{2,0} u_{1,0} - u_{-1,0} u_{-2,0} u_{-3,0})}{\prod_{i=0}^3 (u_{i,0} u_{i-1,0} u_{i-2,0} u_{i-3,0} - \alpha_2 )}, \quad \partial_s u_{0,0} = \frac{u_{0,0}^2 (u_{0,3} u_{0,2} u_{0,1} - u_{0,-1} u_{0,-2} u_{0,-3})}{\prod_{i=0}^3 (u_{0,i} u_{0,i-1} u_{0,i-2} u_{0,i-3} + \alpha_2)},
\end{equation} 
see also \cite{X3}. 

\item[3.] System $\Sigma_3$ is given by the three equations
\begin{subequations}\label{eq:4}
\begin{eqnarray}
&& u_{0,0} \left( u_{1,0} u_{1,1} u_{1,2} + u_{0,1} u_{1,1} u_{1,2} +u_{0,1} u_{0,2} u_{1,2}  +u_{0,1} u_{0,2} u_{0,3}\right) u_{1,3} + \alpha_3,\\
&& u_{0,0} \left(u_{1,0} u_{2,0} u_{2,1}+ u_{1,0} u_{1,1} u_{2,1} + u_{1,0} u_{1,1} u_{1,2} + u_{0,1} u_{1,1} u_{2,1} + u_{0,1} u_{1,1} u_{1,2} + u_{0,1} u_{0,2} u_{1,2} \right) u_{2,2} - \alpha_3 = 0,\\
&& u_{0,0} \left( u_{1,0} u_{2,0} u_{3,0} + u_{1,0} u_{2,0} u_{2,1} +u_{1,0} u_{1,1} u_{2,1}  +u_{0,1} u_{1,1} u_{2,1}\right) u_{3,1} + \alpha_3 = 0= 0,
\end{eqnarray}
\end{subequations}
and its lowest order symmetries are generated by
\begin{equation}\label{eq:4sym}
\partial_t u_{0,0} = \frac{u_{0,0}^2 (u_{4,0} u_{3,0} u_{2,0} u_{1,0} - u_{-1,0} u_{-2,0} u_{-3,0} u_{-4,0})}{\prod_{i=0}^4 (u_{i,0} u_{i-1,0} u_{i-2,0} u_{i-3,0} u_{i-4,0} - \alpha_3 )},\quad 
\partial_s u_{0,0} = \frac{u_{0,0}^2 (u_{0,4} u_{0,3} u_{0,2} u_{0,1} - u_{0,-1} u_{0,-2} u_{0,-3} u_{0,-4})}{\prod_{i=0}^4 (u_{0,i} u_{0,i-1} u_{0,i-2} u_{0,i-3} u_{0,i-4} - \alpha_3)}.
\end{equation}
\end{itemize}

We could have considered lattice paths connecting $(i,0)$ to $(0,j)$ by moving only left or up. This construction leads to consistent systems which actually follow from $\Sigma_N$ by reflecting them over the line $x=0$ (resp. over the line $y=0$), or equivalently by employing the point transformation $u_{k,l}\mapsto u_{k,-l}$ (resp. $u_{k,l} \mapsto u_{-k,l}$). We may also combine the latter transformations with a reciprocal one to derive other equivalent forms of $\Sigma_N$. There is however an interesting construction which employs these two transformations and polynomials $Q_{(N,1)}$, and leads to $N$-quad equations which may be viewed as a deformation of $Q_{(N,1)} + \alpha_{N} =0$. The derivation and some properties of these $N$-quad equations are summarised in the following statement. 

\begin{proposition} \label{prop:Tz-all}
Let  $R_{(i,j)}$ be the polynomial following from $Q_{(i,j)}$ according to
\begin{equation}
R_{(i,j)}  = {\cal{S}}^i \left( \left. Q_{(i,j)}\right|_{u_{k,l} \rightarrow \frac{1}{u_{-k,l}}}\right) \prod_{k=0}^{i} \prod_{l=0}^{j} u_{k,l}= {\cal{T}}^j \left( \left. Q_{(i,j)}\right|_{u_{k,l} \rightarrow \frac{1}{u_{k,-l}}} \right) \prod_{k=0}^{i} \prod_{l=0}^{j} u_{k,l}\,.
\end{equation}
Then the equation
\begin{equation}\label{eq:Tz-all}
Q_{(N,1)} + c_N \,=\,  R_{(N,1)} + \frac{1}{c_N} \prod_{i=0}^{N} u_{i,0} u_{i,1},\quad N=1,2,3,\ldots,
\end{equation}
where $c_N$ is a real constant, admits a hierarchy of symmetries in the first lattice direction. The first member of this hierarchy has order $N+1$ and is generated by
\begin{equation} \label{eq:sym-NT-n}
\partial_t u_{0,0} = u_{0,0} {\cal{S}}^{-N} \left( \frac{\prod_{i=0}^{N+1} {\cal{S}}^{i} \left(Q_{(N-1,0)} - c_N \right)}{\prod_{i=0}^{N-1}  {\cal{S}}^{i} \left(Q_{(N+1,0)} -c_N \right) }\right) (1-{\cal{S}}^{-N-1}) \left(\frac{1}{ Q_{(N+1,0)} -c_N} - \frac{1}{ {\cal{S}}\left(Q_{(N-1,0)} -c_N\right)}\right).
\end{equation}
Moreover, by setting $u \rightarrow u \epsilon^{-1}$, $c_N \rightarrow \alpha_{N} \epsilon^{-N-2}$ and $t \rightarrow t \epsilon^{N+2}$ and considering the limit $\epsilon \rightarrow 0$, equation (\ref{eq:Tz-all}) reduces to $Q_{(N,1)} + \alpha_{N} =0$ and its symmetry (\ref{eq:sym-NT-n}) becomes (\ref{eq:UN}).
\end{proposition}

\begin{proof}
When $N=1$ this is Adler's Tzitzeica equation,
$$u_{0,0} (u_{1,0} + u_{0,1}) u_{1,1} + c_1 = u_{0,0} + u_{1,1} + \frac{u_{0,0} u_{1,0} u_{0,1} u_{1,1}}{c_1}.$$
The corresponding results were first presented in \cite{A}, whereas the relation between (\ref{eq:2}) and (\ref{eq:A2}) was given in \cite{MX}. The existence of the higher order equations was suggested in \cite{AP} (see Remark 4 on page 13) but no explicit formulae were given. The differential-difference equations (\ref{eq:sym-NT-n}) were given in \cite{AP1,AP} along with the Miura transformations 
\begin{equation}\label{eq:NT-n-miura}
{\cal{M}}_{N+1} \,:\quad v_{0,0} = \frac{{\cal{S}} (c_N-Q_{(N-1,0)})}{Q_{(N+1,0)}-c_N}\,u_{N+1,0},\quad {\cal{M}}_0 \,:\quad v_{0,0} = \frac{{\cal{S}} (c_N-Q_{(N-1,0)})}{Q_{(N+1,0)}-c_N}\,u_{0,0}, 
\end{equation}
which map (\ref{eq:sym-NT-n}) to the discrete Sawada-Kotera equation ${\rm{dSK}}^{(1,N)}$
\begin{equation}\label{eq:dSK}
\partial_t v_{0,0} = v_{0,0}^2 \left( \prod_{i=1}^{N+1} v_{i,0} - \prod_{i=1}^{N+1} v_{-i,0}\right) - v_{0,0}\left(\prod_{i=1}^{N} v_{i,0} - \prod_{i=1}^{N} v_{-i,0} \right). 
\end{equation}
The degeneration is a straightforward calculation once the degrees of the polynomials involved are taken into account,  $\deg Q_{(i,j)}=i+j+1$ and  $\deg R_{(i,j)} = i j$.
\end{proof}

Equation (\ref{eq:Tz-all}) and its reflection across the line $x=y$ accompanied by the transformation $c_N \mapsto (-1)^{N+1} c_N$, i.e. 
\begin{equation}\label{eq:Tz-all-b}
Q_{(1,N)} + (-1)^{N+1} c_N \,=\,  R_{(1,N)} + \frac{(-1)^{N+1}}{c_N} \prod_{i=0}^{N} u_{0,i} u_{1,i},\quad N=1,2,3,\ldots,
\end{equation}
may be used as building blocks of other consistent systems. Their construction uses the procedure described in Example \ref{ex:const}. More precisely, it involves the lowest order symmetries of equations (\ref{eq:Tz-all}) and (\ref{eq:Tz-all-b}) along with the requirement that the symmetries must be compatible with every equation of the system. This construction is very involved and the complexity of the calculations increases with $N$. We constructed two new such systems which along with Adler's Tzitzeica equation we denote with $A_1$, $A_2$ and $A_3$, respectively. They depend on a parameter $c_N$ ($N=1,2,3$), and degenerate to $\Sigma_1$, $\Sigma_2$ and $\Sigma_3$, respectively, by setting $u \rightarrow u \epsilon^{-1}$, $c_N \rightarrow \alpha_{N} \epsilon^{-N-2}$  and considering the limit $\epsilon \rightarrow 0$. They satisfy all our three requirements R1--R3 for consistent systems and admit infinite hierarchies of symmetries in both directions, the lowest order of which is $N+1$ and are generated by (\ref{eq:sym-NT-n}).

\begin{itemize}
	\item[1.] System $A_1$, as we have already mentioned, corresponds to Adler's Tzitzeica equation
		\begin{equation}\label{eq:A2}
	u_{0,0} (u_{1,0} + u_{0,1}) u_{1,1} + c_1 = u_{0,0} + u_{1,1} + \frac{u_{0,0} u_{1,0} u_{0,1} u_{1,1}}{c_1},
	\end{equation}
	a well known integrable equation \cite{A}.
	\item[2.] System $A_2$  is constituted by the following two equations.
	\begin{subequations}\label{eq:3T}
		\begin{eqnarray}
		&& u_{0,0}  \left( u_{1,0} u_{1,1} + u_{0,1} u_{1,1} +u_{0,1} u_{0,2} \right) u_{1,2} - c_2  = u_{0,0} u_{0,1} + u_{0,0} u_{1,2} + u_{1,1} u_{1,2} - \frac{u_{0,0} u_{1,0} u_{0,1} u_{1,1} u_{0,2} u_{1,2}}{c_2} \\
		&& u_{0,0} \left( u_{1,0} u_{2,0} + u_{1,0} u_{1,1}  +u_{0,1} u_{1,1}\right) u_{2,1} + c_2  =  u_{0,0} u_{1,0} + u_{0,0} u_{2,1} + u_{1,1} u_{2,1} + \frac{u_{0,0} u_{1,0} u_{2,0} u_{0,1} u_{1,1} u_{2,1}}{c_2} \label{eq:A3a}  
		\end{eqnarray}
	\end{subequations}
	It can be easily verified that this is a consistent system which degenerates to (\ref{eq:3}) as described above (with $N=2$). Its lowest order symmetries in the first direction are generated by
	\begin{equation} \label{eq:A3-n}
	\partial_t u_{0,0} = u_{0,0}\,\frac{\prod_{i=-2}^{1} {\cal{S}}^{i} \left(Q_{(1,0)} - c_2 \right)}{\prod_{i=-2}^{-1}  {\cal{S}}^{i} \left(Q_{(3,0)} -c_2 \right) }\, (1-{\cal{S}}^{-3}) \left(\frac{1}{ Q_{(3,0)} -c_2} - \frac{1}{ {\cal{S}}\left(Q_{(1,0)} -c_2\right)}\right), \quad \left\{\begin{array}{l} Q_{(1,0)} = u_{0,0} u_{1,0} \\ Q_{(3,0)} = u_{0,0} u_{1,0} u_{2,0} u_{3,0} \end{array}\right.,
	\end{equation}
	which is related to the discrete Sawada-Kotera equation ${\rm{dSK}}^{(1,3)}$. Similar considerations hold for the symmetries in the other direction which follow from  (\ref{eq:A3-n}) by applying the changes $u_{\ell,0} \rightarrow u_{0,\ell}$, $c_2 \rightarrow -c_2$ and ${\cal{S}} \rightarrow {\cal{T}}$. 
	\item[3.] System $A_3$ is given by 
	\begin{subequations}\label{eq:4T}
		\begin{eqnarray}
		&& u_{0,0} \left( u_{1,0} u_{2,0} u_{3,0} + u_{1,0} u_{2,0} u_{2,1} +u_{1,0} u_{1,1} u_{2,1}  +u_{0,1} u_{1,1} u_{2,1}\right) u_{3,1} + c_3 = \nonumber \\
		&& \qquad \qquad \qquad \qquad u_{0,0} u_{1,0} u_{2,0} + (u_{0,0} u_{1,1}+ u_{0,0} u_{2,1} + u_{1,1} u_{2,1}) u_{3,1} + \frac{ u_{0,0} u_{1,0} u_{2,0} u_{3,0} u_{0,1} u_{1,1} u_{2,1} u_{3,1} }{c_3} \,, \label{eq:A4a} \\
		&& \nonumber \\
		&& u_{0,0} \left( u_{1,0} u_{1,1} u_{1,2} + u_{0,1} u_{1,1} u_{1,2} +u_{0,1} u_{0,2} u_{1,2}  +u_{0,1} u_{0,2} u_{0,3}\right) u_{1,3} + c_3=  \nonumber  \\
		&& \qquad \qquad \qquad \qquad u_{0,0} u_{0,1} u_{0,2} + (u_{0,0} u_{1,1}+ u_{0,0} u_{1,2} + u_{1,1} u_{1,2}) u_{1,3} + \frac{u_{0,0} u_{0,1} u_{0,2} u_{0,3} u_{1,0} u_{1,1} u_{1,2} u_{1,3}}{c_3}, \\
		&& \nonumber \\
		&& u_{0,0} \left(u_{1,0} u_{2,0} u_{2,1}+ u_{1,0} u_{1,1} u_{2,1} + u_{1,0} u_{1,1} u_{1,2} + u_{0,1} u_{1,1} u_{2,1} + u_{0,1} u_{1,1} u_{1,2} + u_{0,1} u_{0,2} u_{1,2} \right) u_{2,2} - c_3=  \nonumber \\
		&&\qquad \qquad u_{0,0} \left(u_{1,0}+u_{0,1}+u_{2,1}+u_{1,2}\right) u_{2,2}-u_{0,0}-u_{1,1} - u_{2,2} + ({\cal{ST}} +1)\left( u_{0,0} (u_{1,0}+u_{0,1}) u_{1,1}\right) \nonumber \\
		&&\qquad \qquad -\, \frac{u_{0,0} (u_{1,0}+u_{2,1}) (u_{0,1}+u_{1,2}) u_{2,2}  + ({\cal{ST}} +1)\left( u_{0,0}  u_{1,0} u_{0,1} u_{1,1}\right)}{c_3}\nonumber \\
		&& \quad \quad+\, \frac{u_{0,0} \left((u_{1,0} (u_{2,0}+u_{1,1}) + u_{0,1} (u_{0,2}+u_{1,1})) u_{2,1} u_{1,2} + u_{1,0} u_{0,1} \left( (u_{2,0}+u_{1,1}) u_{2,1} + (u_{0,2}+ u_{1,1}) u_{1,2}\right) \right) u_{2,2}}{c_3} \nonumber\\
		&& \qquad \qquad -\, \frac{ u_{0,0} u_{1,0} u_{0,1} u_{2,1} u_{1,2} u_{2,2} \left( u_{2,0} u_{0,2} + u_{1,1} (u_{2,0}+u_{0,2})\right)}{c_3}   +\frac{u_{0,0} u_{1,0} u_{2,0} u_{0,1} u_{1,1} u_{2,1} u_{0,2} u_{1,2} u_{2,2}}{c_3^2} . 
		\end{eqnarray}	\end{subequations}
		It degenerates to (\ref{eq:4}) and its lowest order symmetries in the first direction are generated by
		\begin{equation} \label{eq:A4-n}
  		\partial_t u_{0,0} = u_{0,0}  \frac{\prod_{i=-3}^{1} {\cal{S}}^{i} \left(Q_{(2,0)} - c_3 \right)}{\prod_{i=-3}^{-1}  {\cal{S}}^{i} \left(Q_{(4,0)} -c_3 \right) } (1-{\cal{S}}^{-4}) \left(\frac{1}{ Q_{(4,0)} -c_3} - \frac{1}{ {\cal{S}}\left(Q_{(2,0)} -c_3\right)}\right), \quad \left\{\begin{array}{l} Q_{(2,0)} = u_{0,0} u_{1,0} u_{2,0} \\ Q_{(4,0)} = u_{0,0} u_{1,0} u_{2,0} u_{3,0} u_{4,0} \end{array}\right.,
		\end{equation}
		which is related to the discrete Sawada-Kotera equation ${\rm{dSK}}^{(1,4)}$. The symmetries in the other direction follow from (\ref{eq:A4-n}) by applying the changes $u_{\ell,0} \rightarrow u_{0,\ell}$ and ${\cal{S}} \rightarrow {\cal{T}}$. 
\end{itemize}

\section{Conclusions \& Discussion}

We considered $N$-th order overdetermined systems of difference equations which are consistent and integrable according to our requirements and definitions in Section \ref{sec:def}. We demonstrated how such systems follow from known lower order integrable systems and presented two new hierarchies. The first one was constructed using lattice paths whereas the second hierarchy can be interpreted as a deformation of the former. In particular the first members of these hierarchies coincide with the quad equation (\ref{eq:2}) introduced in \cite{MX} and Adler's Tzitzeica equation (\ref{eq:A2}) studied in \cite{A}, respectively. In this way we have shown that these two equations are not isolated but they are the lowest order members of two hierarchies of consistent systems denoted here with $\Sigma$ and $A$, respectively. Systems $\Sigma_N$ can be constructed for any order $N$ but, due to computational limitations, we were able to construct only the first three members of the $A$ hierarchy.

There are a lot of interesting questions about consistent systems. It is very well known that multidimensional consistency is a strong integrability property closely related to other integrability aspects,  e.g. Lax pairs and B{\"a}cklund transformations. However it is not clear if the type of consistency considered here can be employed in a similar way. Most of the well known integrable equations also fit into the framework of direct linearization or Kac-Moody algebras or can be derived as reductions of discrete KP equations. Could overdetermined consistent systems be derived in any of these ways? On the other hand from the examples we presented it seems that there exists a relation between consistency and symmetries of $N$-quad equations. It would be interesting to explore this connection further in order to understand the structure of symmetries of $N$-quad equations but also to derive integrability conditions for consistent systems.

\end{document}